\documentclass[11pt, a4paper, copyright, gdm]{google}
\usepackage[authoryear, compress, round]{natbib}

\usepackage{mathtools}
\usepackage{physics}
\usepackage{cleveref}
\usepackage{xspace}
\usepackage{amsthm}
\usepackage{tcolorbox}
\usepackage{caption}
\newcommand{\method}{{AlphaEvolve}\xspace}

\theoremstyle{definition}

\theoremstyle{remark}

\theoremstyle{plain}

\newtheorem{lemma}{Lemma}
\newtheorem{theorem}{Theorem}

\DeclareMathOperator{\Split}{Split}
\newcommand{\defn}[1]{{\emph{{#1}}}}

\newcommand{\Shapes}{\mathcal{S}}
\DeclarePairedDelimiterX\mysetbase[2]{\lbrace}{\rbrace}{#1\;\delimsize\vert\;#2}
\NewDocumentCommand{\myset}{sO{}m m}{%
  \IfBooleanTF{#1}% Check if starred version
    {\mysetbase*{#3}{#4}}% Automatic sizing
    {\mysetbase[#2]{#3}{#4}}% Manual sizing
}
\newcommand{\Z}{\mathbb{Z}}
\newcommand{\defeq}{\coloneqq}
\DeclareMathOperator{\supp}{supp}
\DeclarePairedDelimiter{\bk}{(}{)}
\DeclarePairedDelimiter{\BK}{\{}{\}}
\newcommand{\numberthis}{\addtocounter{equation}{1}\tag{\theequation}}
\renewcommand{\ip}{\text{\ttfamily +}}  % coordinate > 0
\newcommand{\is}{\text{\ttfamily *}}  % coordinate unrestricted
\newcommand{\il}{\text{\ttfamily <}}  % coordinate below its maximum
\newcommand{\eps}{\varepsilon}
\DeclarePairedDelimiter{\absolute}{\lvert}{\rvert}
\DeclarePairedDelimiter{\angbk}{\langle}{\rangle}
\uselogo{} 

\title{Improving the matrix multiplication exponent with modern optimization and AlphaEvolve}

\correspondingauthor{edupont@google.com}

\author[*,1]{Emilien Dupont}
\author[*,1]{Marvin Eisenberger}
\author[*,1]{Borislav Kozlovskii}
\author[*,1]{Abbas Mehrabian}
\author[*,1]{Francisco J.\ R.\ Ruiz}
\author[*,1]{Abigail See}
\author[*,2]{Renfei Zhou}
\author[3]{Josh Alman}
\author[4]{Virginia Vassilevska Williams}
\author[1]{Matej Balog}

\affil[*]{Equal contribution in alphabetical order}
\affil[1]{\thepa{}{}}
\affil[2]{Carnegie Mellon University}
\affil[3]{Columbia University}
\affil[4]{MIT}

\begin{abstract}
The current best bounds on the matrix multiplication exponent $\omega$ are obtained through a refinement of the laser method called combination loss analysis \citep{duan2022faster, williams2024new, alman2025more}. In this note, we address the optimization problem at the core of this approach and propose several improvements. First, we reformulate the optimization problem allowing us to solve it in a larger setting than was previously possible. Second, we leverage recent advances in machine learning to design a new optimization algorithm for this problem. Finally, we refine the resulting optimization algorithm with \method. Our combined approach yields an upper bound of $\omega < 2.371177$, improving the previous best bound of $2.371339$.
\end{abstract}

\begin{document}

\maketitle

\section{Introduction}

From accelerating machine learning computations to enabling realistic computer graphics, matrix multiplication is a fundamental operation underpinning critical applications in computer science. Despite its prominence, the computational complexity of matrix multiplication---the number of arithmetic operations needed to multiply large matrices---is unknown, and determining it is a major open question in theoretical computer science \citep{blaser2013fast}.

The pioneering work of Strassen \citep{strassen1969gaussian} showed that two $n\times n$ matrices can be multiplied in sub-cubic time---specifically, $\mathcal{O}(n^{\omega + o(1)})$ operations for $\omega < 2.81$---spurring a line of work attempting to further reduce the complexity exponent $\omega$ \citep{pan1978strassen,bini1979order,schonhage1981partial,romani1982some,coppersmith1981asymptotic,strassen1986asymptotic,coppersmith1990matrix,stothers2010complexity,williams2012multiplying,legall2014algebraic,alman2024refined,duan2022faster,williams2024new,alman2025more}. All improvements in the past 40 years rely on the \emph{laser method}, a mathematical technique to indirectly design matrix multiplication algorithms. The current best bound is achieved by a refinement of the laser method called \emph{combination loss analysis} \citep{duan2022faster}, a technique that requires solving a non-convex optimization problem as part of the computer-assisted proof, and yields $\omega < 2.371339$ \citep{alman2025more}.

Here, we improve the bound to $\omega < 2.371\mathbf{177}$ using a two-step approach. First, we leverage recent advances in machine learning and adjacent areas to address the non-convex optimization problem of combination loss analysis using a gradient descent approach; this alone improves the previous state-of-the-art (SOTA) bound by $\approx 0.97 \times 10^{-4}$. Second, we use \method \citep{novikov2025alphaevolve} to improve our optimization algorithm; this raises the improvement over the SOTA to $\approx 1.62 \times 10^{-4}$.

\begin{table}[b]
\centering
\begin{tabular}{l|l}
\cite{duan2022faster} & $2.371\mathbf{866}$   \\ \hline
\cite{williams2024new} & $2.371\mathbf{552}$  \\ \hline
\cite{alman2025more} & $2.371\mathbf{339}$  \\ \hline
This note & $2.371\mathbf{177}$ \\
\end{tabular}
\caption{Recent improvements to $\omega$.}
% \label{tab:omega_values}
\end{table}

The optimization problem at the core of combination loss analysis is formulated in~\cite{alman2025more}, where it was also shown that any feasible solution provides an upper bound on $\omega$. To achieve our new bound, we target a slightly different optimization problem. Specifically, combination loss analysis has a parameter---the \emph{maximum recursion level}, denoted by $\ell^*$---that introduces a trade-off between the complexity of the optimization (which grows doubly exponentially in $\ell^*$) and the best possible bound on $\omega$ it can achieve. The previous SOTA bound was found with $\ell^*=3$ \citep{alman2025more}. In contrast, our gradient-based optimization, implemented in Jax \citep{jax2018github}, allows for hardware parallelization and can handle $\ell^*=4$ (the number of optimizable parameters increases from approximately 25k to 7 million when moving from $\ell^* = 3$ to $\ell^* = 4$).

In this note, we describe the full optimization problem from \cite{alman2025more} in detail, explain how we numerically solved it and applied \method to it, and finally show how we rigorously certified the resulting omega bound.

\section{Optimization problem}

The high-level structure of the optimization problem can be captured by a rooted tree, where every node is associated with a collection of optimizable parameters.
Integers $q \ge 1$ and $\ell^* \ge 2$ are fixed as hyperparameters.
Intuitively, the tree structure describes a recursive way to decompose the tensor
$CW_q^{\otimes 2^{\ell^*}}$ into smaller tensors, where $CW_q$ is the Coppersmith-Winograd tensor~\citep{coppersmith1990matrix}; see \cite{alman2025more} for more details about this correspondence. We start by defining the tree structure, and will then define the optimizable parameters associated to its nodes.

\subsection{Tree structure}
\label{subsec:recursive-tree-structure}
Denote $[k]\defeq\{1,\dots,k\}$.
Each non-root node $T$ in the tree is associated with a \emph{level}, a \emph{shape} $s_T^{}$, and a \emph{region} $r_T^{}$, defined as follows.
\begin{itemize}
  \item The \defn{level} of a non-root node is a positive integer $\ell \in \{2,3,\dots,\ell^*\}$, describing its depth in the tree; higher is closer to the root. Direct children of the root node have level $\ell^*$, which is a hyperparameter fixed in advance.
  \item A level-$\ell$ \defn{shape} is a triple $s = (s_X^{}, s_Y^{}, s_Z^{})$, where $s_X^{}, s_Y^{}, s_Z^{}$ are non-negative integers summing to $2^\ell$. We use
        \[
          \Shapes_\ell^{} \coloneqq \myset*{(i,j,k) \in \Z_{\ge 0}^3}{i+j+k=2^\ell}
        \]
        to denote the set of all level-$\ell$ shapes. We use \defn{dimensions} $X, Y, Z$ to refer to the indices of the three coordinates in a shape.
  \item There are six \defn{regions}, indexed by an integer $r \in [6]$. Each region is associated with $\pi_r$, the $r$-th permutation over symbols $\{X, Y, Z\}$ in lexicographic order.
\end{itemize}
We call a non-root node $T$ a \defn{positive-shape node} if \emph{all} coordinates in its shape $s_T^{}$ are positive. Otherwise it is a \defn{zero-shape node}. The root node is denoted by $G$ and does not have the aforementioned attributes---level, shape, and region.

Next, we define the tree structure by describing the child nodes for different types of nodes.

\begin{itemize}
  \item \textbf{Root:} For every region $r \in [6]$ and level-$\ell^*$ shape $s \in \Shapes_{\ell^*}^{}$, the root $G$ has a child at level $\ell^*$ with shape $s$ and region $r$, denoted as $G[s, r]$.
  \item \textbf{Positive-shape node:} For every level-$\ell$ shape $s = (s_X^{}, s_Y^{}, s_Z^{})$, we define
        \[
          \Split(s) \defeq \myset*{u \in \Shapes_{\ell-1}}{0 \le u_X^{} \le s_X^{}, \; 0 \le u_Y^{} \le s_Y^{}, \; 0 \le u_Z^{} \le s_Z^{}}.
        \]
        Fixing a positive-shape node $T$ at level $\ell \ge 3$, for every region $r \in [6]$ and level-$(\ell - 1)$ shape $u \in \Split(s_T^{})$, the node $T$ has a child at level $\ell - 1$ with shape $u$ and region $r$, denoted as $T[u, r]$. Note that the region index $r$ of the child can be different from that of $T$ itself.
  \item \textbf{Zero-shape nodes} and \textbf{level-2 positive-shape nodes} do not have child nodes. We call them the \emph{leaves}.
\end{itemize}

\subsection{Optimizable parameters}
\label{subsec:optimizable-parameters}
Next we list the free variables of the optimization problem, grouped by the associated node on the tree. For any finite set $D$, we use $\Delta(D)$ to denote the simplex of probability distributions on $D$:
\[
\Delta(D) \defeq \myset*{p: D \rightarrow [0,1]}{\sum_{x\in D} p(x) = 1}
\]

\paragraph{Root node.}
The optimizable parameters associated with the root node $G$ are:
\begin{itemize}
  \item A distribution $A_G^{} = (A_G^{(1)}, \dots, A_G^{(6)}) \in \Delta([6])$ over the six regions;
  \item For $r \in [6]$, a distribution $\alpha_G^{(r)} \in \Delta(\Shapes_{\ell^*}^{})$ over all level-$\ell^*$ shapes.
\end{itemize}

\paragraph{Positive-shape node.}
Each level-$\ell$ positive-shape node $T$ with $\ell \ge 3$ is associated with the following parameters:
\begin{itemize}
  \item A distribution $A_T^{}=(A_T^{(1)},\ldots,A_T^{(6)}) \in \Delta([6])$ over the six regions;
  \item For $r \in [6]$, a distribution $\alpha_T^{(r)}\in\Delta(\Split(s_T^{}))$, where $s_T^{}$ is the shape of $T$.
\end{itemize}

\paragraph{Zero-shape node.}
For a level $\ell$ and an integer $a$ with $0 \le a \le 2^\ell$, we define
\[
  \mathcal C_{\ell,a}^{}
  \coloneqq
  \myset*{L\in\{0,1,2\}^{2^{\ell-1}}}{\sum_{p=1}^{2^{\ell-1}} L_p=a}.
\]
This is the set of length-$2^{\ell-1}$ vectors over $\{0,1,2\}$ whose entries sum to $a$.  A distribution over $\mathcal C_{\ell, a}^{}$ is called a level-$\ell$ \defn{complete split distribution} of $a$.

For every level-$\ell$ zero-shape node $T$, let $W \in \{X, Y, Z\}$ be the first dimension where $s_{T,W}^{}$ is nonzero. $T$ is then associated with a complete split distribution $\beta_{T,W}^{} \in \Delta(\mathcal C_{\ell, s_{T,W}^{}}^{})$.

\paragraph{Level-2 node.}
The only remaining  nodes are level-2 nodes with shapes $(1,1,2)$, $(1,2,1)$, or $(2,1,1)$, which are the only valid strictly positive shapes for level 2, since their coordinates must sum to $2^2=4$. Each such node is associated with a scalar $\mu_{T}^{} \in [0, 1/2]$.

\subsection{Derived quantities}
\label{subsec:derived-quantities}

With the free variables fixed, all remaining quantities can be computed deterministically.
In the following, all logarithms and entropies are in base two.  For a distribution $\rho$ with finite support, we denote its entropy by
\[
  H(\rho) \defeq -\sum_{x\in\supp(\rho)}\rho(x) \log \rho(x).
\]
We also use $H(p_1, p_2, \dots, p_k)$ to denote the entropy of a distribution over $[k]$ with probability masses $p_1, \dots, p_k$.
If $D$ is a set of shapes and $\rho \in \Delta(D)$, let $\rho_W^{}$ denote the marginal of $\rho$ in coordinate $W$, namely, $\rho_W(w) \defeq \displaystyle\sum_{a \in D : a_W = w} \rho(a)$.
Define
\[
  H_D^{\max}(\rho)
  \defeq
  \sup_{\substack{\rho'\in\Delta(D)\\
      \rho'_W=\rho_W^{}\text{ for }W\in\{X,Y,Z\}}}
  H(\rho'),
  \numberthis \label{eq:def-hmax}
\]
as well as the penalty notion
\[
  P_{D}^{}\bk*{\rho}
  \defeq
  H_{D}^{\max}\bk*{\rho}
  - H\bk*{\rho}.
\]

\paragraph{Masses.}
Each non-zero-shape node $T$ of the tree has a real number $m_T^{} \in [0, 1]$ associated with it, called its \defn{mass}. The masses of non-root nodes are computed top-down as follows:
\begin{itemize}
  \item The root $G$ has $m_G^{} = 1$.
  \item For the children of the root, we set
        \[
          m_{G[s,r]}^{} \defeq A_G^{(r)} \cdot \alpha_G^{(r)}(s),
          \qquad \forall\, s \in \Shapes_{\ell^*}^{},\ r \in [6].
        \]
  \item For any positive-shape node $T$ of level $\ell \ge 3$, we set
        \[
          m_{T[u,r]}^{}=m_T^{} \cdot A_T^{(r)} \cdot \left(\alpha_T^{(r)}(u)+\alpha_T^{(r)}(s_T^{} \!-\! u)\right),
          \qquad \forall\, u \in \Split(s_T^{}),\ r \in [6].
        \]
\end{itemize}

\paragraph{Complete split distributions.}
Every level-$\ell$ non-root node $T$ carries, for each dimension $W \in \{X,Y,Z\}$, a complete split distribution $\beta_{T,W}^{} \in \Delta\bk[\big]{\mathcal C_{\ell, s_{T,W}^{}}^{}}$.  These distributions are calculated as follows.

Let $\vec{0} \defeq (0, 0, \dots, 0)$ and $\vec{2} \defeq (2, 2, \dots, 2)$ denote vectors of length $2^{\ell - 1}$.
For a zero-shape node $T$, let $W_0^{} \in \BK{X, Y, Z}$ be the first zero coordinate of $s_T^{}$, $W_1^{}$ be the first nonzero coordinate of $s_T^{}$, and $W_2^{}$ be the other coordinate. $\beta_{T, W_0}^{}$ is the point mass distribution at the length-$2^{\ell - 1}$ vector $\vec{0}$; $\beta_{T, W_1}^{}$ was defined as optimizable parameters in \cref{subsec:optimizable-parameters}; $\beta_{T, W_2}^{} \defeq \beta^{\vee}_{T, W_1}$, where for any complete split distribution $\beta$, we define
\[
  \beta^{\vee}(\vec{2} \!-\! L) \defeq \beta(L), \qquad \forall\, L \in {\supp(\beta)}.
\]
(Notice, in particular, that if the components of $L$ sum to $s_{T, W_1}$, then the components of $\vec{2} - L$ sum to $2 \cdot 2^{\ell-1} - s_{T, W_1} = s_{T, W_2}$, keeping the mapping validly within the correct domain.)

For a positive-shape node $T$ at level $\ell \ge 3$, we define
\[
  \beta_{T,W}^{(r)}
  \defeq \sum_{u\in\Split(s_T^{})}\alpha_T^{(r)}(u) \cdot
  \bigl(\beta_{T[u,r],W}^{}\times\beta_{T[s_T^{}-u,r],W}^{}\bigr),
  \qquad
  \forall\, r \in [6],
\]
where $\times$ denotes the Cartesian product of complete split distributions; that is, if $L_{\text{left}}$ and $L_{\text{right}}$ are two sequences of length $2^{\ell-2}$ in supports of 
$\beta_{\text{left}} \defeq \beta_{T[u,r],W}$
and $\beta_{\text{right}}\defeq\beta_{T[s_T^{}-u,r],W}$, respectively,
then we form a new sequence $L$, of length $2^{\ell-1}$, by concatenating $L_{\text{left}}$ and $L_{\text{right}}$,
and letting
\[
  (\beta_{\text{left}} \times \beta_{\text{right}})(L) \defeq \beta_{\text{left}}(L_{\text{left}}) \cdot \beta_{\text{right}}(L_{\text{right}}).
\]
Then, the complete split distributions of $T$ are calculated by
\[\beta_{T,W}^{} \defeq \sum_{r=1}^6 A_T^{(r)}\beta_{T,W}^{(r)}.\]

The only remaining case is a level-2 positive-shape node, which must have shape $(1,1,2)$, $(2,1,1)$, or $(1,2,1)$. For a node $T$ with shape $(1,1,2)$, recall that $\mu_T^{} \in [0, 1/2]$ is the only optimizable parameter associated with $T$. We let $\delta_{a,b}$ denote the point mass at $(a, b) \in \BK{0,1,2}^2$ and set
\[
  \beta_{T,X}^{} = \beta_{T,Y}^{} = \frac{1}{2}\delta_{0,1}+\frac{1}{2}\delta_{1,0},
  \qquad
  \beta_{T,Z}^{}=\mu_T^{}\delta_{0,2}+\mu_T^{}\delta_{2,0}+(1-2\mu_T^{})\delta_{1,1}. \numberthis \label{eq:def-csd-112}
\]
Similarly, for a node $T$ with shape $(2,1,1)$, we set
\[
  \beta_{T,Y}^{} = \beta_{T,Z}^{} = \frac{1}{2}\delta_{0,1}+\frac{1}{2}\delta_{1,0},
  \qquad
  \beta_{T,X}^{}=\mu_T^{}\delta_{0,2}+\mu_T^{}\delta_{2,0}+(1-2\mu_T^{})\delta_{1,1}, 
\]
and for a node $T$ with shape $(1,2,1)$, we set
\[
  \beta_{T,X}^{} = \beta_{T,Z}^{} = \frac{1}{2}\delta_{0,1}+\frac{1}{2}\delta_{1,0},
  \qquad
  \beta_{T,Y}^{}=\mu_T^{}\delta_{0,2}+\mu_T^{}\delta_{2,0}+(1-2\mu_T^{})\delta_{1,1}.
\]

\paragraph{Retained exponent at root node.}
Next, we define how to calculate the \defn{retained exponent} $E_G^{}$ associated with the root $G$, which will be obtained from intermediate quantities $E_G^{(r)}$ for regions $r \in [6]$.
We start by introducing how to compute $E^{(1)}_G$ for region $r = 1$, where $\pi_1$ is the identity permutation.
In all summations below, $s$ ranges over $\Shapes_{\ell^*}^{}$.
We will use the symbols $\is,\ip,$ and $\il$ in our notation below as mnemonic placeholders corresponding to the dimensions, indicating the meanings ``unconstrained / any value'', ``strictly positive'', and ``strictly less than the parent's coordinate'', respectively. 
We define
\[
  \eta_{G,Y}^{(1)}
  \defeq
  \sum_{s \,:\, s_Z^{}=0} \alpha_G^{(1)}(s) \cdot H\bk*{\beta_{G[s,1],Y}^{}}
  +\sum_{j=0}^{2^{\ell^*}}\alpha_G^{(1)}(\is,j,\ip) \cdot H\bk*{\bar\beta_{G,Y,\is,j,\ip}^{(1)}},
  \numberthis \label{eq:global-eta-Y}
\]
where
\[
  \alpha_G^{(1)}(\is,j,\ip)
  \defeq
  \sum_{s \,:\, s_Y^{}=j,\; s_Z^{}>0}\alpha_G^{(1)}(s),
  \qquad
  \bar\beta_{G,Y,\is,j,\ip}^{(1)}
  \defeq
  \frac{1}{\alpha_G^{(1)}(\is,j,\ip)}
  \sum_{s \,:\, s_Y^{}=j,\; s_Z^{} > 0} \alpha_G^{(1)}(s) \cdot \beta_{G[s,1],Y}^{}.
\]
We then define
\[
  \eta_{G,Z}^{(1)}
  \defeq
  \sum_{s \,:\, s_X^{}=0\ \mathrm{or}\ s_Y^{}=0} \alpha_G^{(1)}(s) \cdot H\bk*{\beta_{G[s,1],Z}^{}}
  +\sum_{k=0}^{2^{\ell^*}}\alpha_G^{(1)}(\ip,\ip,k) \cdot H\bk*{\bar\beta_{G,Z,\ip,\ip,k}^{(1)}},
  \numberthis \label{eq:global-eta-Z}
\]
where
\[
  \alpha_G^{(1)}(\ip,\ip,k)
  \defeq
  \sum_{s \,:\, s_X^{}>0,\; s_Y^{}>0,\; s_Z^{}=k}\alpha_G^{(1)}(s),
  \qquad
  \bar\beta_{G,Z,\ip,\ip,k}^{(1)}
  \defeq
  \frac{1}{\alpha_G^{(1)}(\ip,\ip,k)}
  \sum_{s \,:\, s_X^{}>0,\; s_Y^{}>0,\; s_Z^{}=k} \alpha_G^{(1)}(s) \cdot \beta_{G[s,1],Z}^{}.
\]
To avoid the issue of denominators $\alpha_G^{(1)}(\cdot)$ being zero, we regard $0 \cdot \textsf{undefined} \defeq 0$, so that \Cref{eq:global-eta-Y} and \Cref{eq:global-eta-Z} are always well-defined.

The quantity $E_G^{(1)}$ is then given by
\[
  E_G^{(1)}
  \defeq
  \min\BK*{
  H\bk*{\bk[\big]{\alpha_G^{(1)}}_X^{}} - P_{\Shapes_{\ell^*}}\bk*{\alpha_G^{(1)}},\;
  H\bk*{\bar\beta_{G,Y,\is,\is,\is}^{(1)}} - \eta_{G,Y}^{(1)},\;
  H\bk*{\bar\beta_{G,Z,\is,\is,\is}^{(1)}} - \eta_{G,Z}^{(1)}
  },
  \numberthis \label{eq:global-E-region-1}
\]
where
%\[P_{\Shapes_{\ell^*}}\bk*{\alpha_G^{(1)}} \defeq H_{\Shapes_{\ell^*}}^{\max}\bk*{\alpha_G^{(1)}} - H\bk*{\alpha_G^{(1)}},\] and
\[
  \bar\beta_{G,W,\is,\is,\is}^{(1)}
  \defeq
  \sum_{s \in \Shapes_{\ell^*}^{}} \alpha_G^{(1)}(s) \cdot \beta_{G[s,1],W}^{},
  \qquad \forall\, W \in \{X,Y,Z\}.
\]

For $r = 2, \dots, 6$, we define $E_G^{(r)}$ by applying the same formula \Cref{eq:global-E-region-1} after relabeling the coordinates $X,Y,Z$ as $\pi_r^{}(X),\pi_r^{}(Y),\pi_r^{}(Z)$, and replacing the region index $1$ with $r$. Finally, we set
\[
  E_G^{} \defeq \sum_{r=1}^6 A_G^{(r)} \cdot E_G^{(r)}.
\]

\paragraph{Retained exponents at level $\ell \ge 3$.}
For each positive-shape node $T$ at level $\ell \ge 3$, we will calculate intermediate quantities $E_{T,W}^{(r)}$ for each {region} $r \in [6]$ and $W \in \BK{X,Y,Z}$. These quantities will later be aggregated to form the \emph{retained exponent} of level $\ell$.
% For each \defn{region} $r \in [6]$ and dimension $W \in \BK{X,Y,Z}$, we will calculate a quantity $E_{T,W}^{(r)}$.
Note that the region index $r$ in the calculation can be different from the region index $r_{T}^{}$ of node $T$ itself. As in the above, we start by defining the quantities for region $r = 1$ where $\pi_1^{}$ is the identity permutation. In the summations below, $u$ ranges over $\Split(s_T^{})$.

% Copyediting note: When we write s_T - u or other formulas representing the complement shape, use s_T \!-\! u for reduced spacing, except in subscripts.

We first define
\[
  \begin{aligned}
    \eta_{T,Y}^{(1)}
     & \defeq
    \sum_{u \,:\, u_Z^{} = 0}
    \Bigl(\alpha_T^{(1)}(u) + \alpha_T^{(1)}(s_T^{} \!-\! u)\Bigr) \cdot
    H\bk*{\beta_{T[u,1],Y}^{}} \\
     & \quad+
    \sum_{j=0}^{\min\BK{s_{T,Y}^{},\, 2^{\ell - 1}}}
    \Bigl(
    \alpha_T^{(1)}(\is,j,\ip)
    + \alpha_T^{(1)}(\is,\, s_{T,Y}^{} \!-\! j,\, \il)
    \Bigr) \cdot
    H\bk*{\bar\beta_{T,Y,\is,j,\ip}^{(1)}},
  \end{aligned}
  \numberthis \label{eq:positive-shape-eta-Y}
\]
where
\[
  \begin{aligned}
    \alpha_T^{(1)}(\is,j,\ip)
    \defeq
    \sum_{u \,:\, u_Y^{} = j,\; u_Z^{} > 0}
    \alpha_T^{(1)}(u),
    \qquad
    \alpha_T^{(1)}(\is,\, s_{T,Y}^{} \!-\! j,\, \il)
    \defeq
    \sum_{\substack{u \,:\, u_Y^{} = s_{T,Y}^{} - j, \\
    u_Z^{} < s_{T,Z}^{}}}
    \alpha_T^{(1)}(u),
  \end{aligned}
\]
\[
  \bar\beta_{T,Y,\is,j,\ip}^{(1)}
  \defeq
  \frac{1}{\alpha_T^{(1)}(\is,j,\ip) + \alpha_T^{(1)}(\is,\, s_{T,Y}^{} \!-\! j,\, \il)}
  \cdot
  \sum_{u \,:\, u_Y^{} = j,\; u_Z^{} > 0}
  \Bigl(\alpha_T^{(1)}(u) + \alpha_T^{(1)}(s_T^{} \!-\! u)\Bigr) \cdot \beta_{T[u,1],Y}^{}.
\]
We also define
\[
  \begin{aligned}
    \eta_{T,Z}^{(1)}
     & \defeq
    \sum_{u \,:\, u_X^{} = 0\ \mathrm{or}\ u_Y^{} = 0}
    \Bigl(\alpha_T^{(1)}(u) + \alpha_T^{(1)}(s_T^{} - u)\Bigr) \cdot
    H\bk*{\beta_{T[u,1],Z}^{}} \\
     & \quad+
    \sum_{k=0}^{\min\BK{s_{T,Z}^{},\, 2^{\ell-1}}}
    \Bigl(
    \alpha_T^{(1)}(\ip,\ip,k)
    + \alpha_T^{(1)}(\il,\, \il,\, s_{T,Z}^{} \!-\! k)
    \Bigr) \cdot
    H\bk*{\bar\beta_{T,Z,\ip,\ip,k}^{(1)}},
  \end{aligned}
  \numberthis \label{eq:positive-shape-eta-Z}
\]
where
\[
  \begin{aligned}
    \alpha_T^{(1)}(\ip,\ip,k)
    \defeq
    \sum_{u \,:\, u_X^{} > 0,\; u_Y^{} > 0,\; u_Z^{} = k}
    \alpha_T^{(1)}(u),
    \qquad
    \alpha_T^{(1)}(\il,\, \il,\, s_{T,Z}^{} \!-\! k)
    \defeq
    \sum_{\substack{u \,:\, u_X^{} < s_{T,X}^{},\; u_Y^{} < s_{T,Y}^{}, \\
    u_Z^{} = s_{T,Z}^{} - k}}
    \alpha_T^{(1)}(u),
  \end{aligned}
\]
\[
  \bar\beta_{T,Z,\ip,\ip,k}^{(1)}
  \defeq
  \frac{1}{\alpha_T^{(1)}(\ip,\ip,k) + \alpha_T^{(1)}(\il,\, \il,\, s_{T,Z}^{} \!-\! k)}
  \cdot
  \sum_{u \,:\, u_X^{} > 0,\; u_Y^{} > 0,\; u_Z^{} = k}
  \Bigl(\alpha_T^{(1)}(u) + \alpha_T^{(1)}(s_T^{} \!-\! u)\Bigr) \cdot \beta_{T[u,1],Z}^{}.
\]
Again, as we define $0 \cdot \textsf{undefined} \defeq 0$, \Cref{eq:positive-shape-eta-Y} and \Cref{eq:positive-shape-eta-Z} are well-defined even when some denominators $\alpha_T^{(1)}$ are zeros.

Then, $E^{(1)}_{T,W}$ are given by
\[
  \begin{aligned}
    E_{T,X}^{(1)}
     & \defeq
    m_T^{} A_T^{(1)} \cdot
    \bk[\Big]{
      H\bk*{\bk[\big]{\alpha_{T}^{(1)}}_{X}^{}}
      - P_{\Split(s_T^{})}^{}\bk*{\alpha_T^{(1)}}
    },        \\
    E_{T,Y}^{(1)}
     & \defeq
    m_T^{} A_T^{(1)} \cdot
    \bk[\Big]{
      H\bk*{\beta_{T,Y}^{(1)}}
      - \eta_{T,Y}^{(1)}
    },        \\
    E_{T,Z}^{(1)}
     & \defeq
    m_T^{} A_T^{(1)} \cdot
    \bk[\Big]{
      H\bk*{\beta_{T,Z}^{(1)}}
      - \eta_{T,Z}^{(1)}
    }.
  \end{aligned}
  \numberthis \label{eq:E-positive-shape}
\]
% where $\gamma_{T,W}^{(1)}$ for $W \in \{X,Y,Z\}$ is a distribution over $\BK{ (a, b) \in \N^{2} \mid a + b = s_{T,W}^{}}$ given by
% \[
%   \gamma_{T,W}^{(1)}(a,b)
%   \defeq
%   \sum_{u \,:\, u_W^{} = a,\; \bk{s_T^{} - u}_W^{} = b}
%     \alpha_T^{(1)}(u),
% \]
% and
% where
% \[
%   P_{\Split(s_T^{})}^{}\bk*{\alpha_T^{(1)}}
%   \defeq
%   H_{\Split(s_T^{})}^{\max}\bk*{\alpha_T^{(1)}}
%   - H\bk*{\alpha_T^{(1)}}.
% \]

For regions $r = 2, \dots, 6$, we define $E_{T,\pi_r^{}(X)}^{(r)}$, $E_{T,\pi_r^{}(Y)}^{(r)}$, and $E_{T,\pi_r^{}(Z)}^{(r)}$ by applying
\Cref{eq:positive-shape-eta-Y}, \Cref{eq:positive-shape-eta-Z}, and \Cref{eq:E-positive-shape} after relabeling coordinates $X,Y,Z$ as $\pi_r^{}(X),\pi_r^{}(Y),\pi_r^{}(Z)$ and replacing region index $1$ by $r$.

To calculate the \emph{retained exponent} for level $\ell \in [3, \ell^*]$, we let $\mathcal T_\ell^+$ denote the set of positive-shape nodes at level $\ell$. Then, the retained exponent at level $\ell$ is given by
\[
  E_\ell^{}
  \defeq
  \sum_{r = 1}^6
  \min\BK*{
    \sum_{T \in \mathcal T_\ell^+} E_{T,X}^{(r)}, \;
    \sum_{T \in \mathcal T_\ell^+} E_{T,Y}^{(r)}, \;
    \sum_{T \in \mathcal T_\ell^+} E_{T,Z}^{(r)}
  }.
\]

\paragraph{Retained exponent at level 2.}

For a positive level-$2$ node $T$ of shape $(1,1,2)$, we define
\[
  \bk*{E_{T,X}^{}, E_{T,Y}^{}, E_{T,Z}^{}}
  \defeq
  m_T^{} \cdot
  \bk*{1,\; 1,\; H\bk*{\mu_T^{}, \mu_T^{}, 1 - 2\mu_T^{}}}.
  \numberthis \label{eq:E-112}
\]
% $E_{T,X}^{}, E_{T,Y}^{}, E_{T,Z}^{}$ for nodes $T$ with shapes $(2,1,1)$ and $(1,2,1)$ are obtained from the right-hand side of \Cref{eq:E-112} by cyclically shifting to the right-hand side once and twice, respectively.
Similarly, for a node $T$ of shape $(2,1,1)$, we define
\[
  \bk*{E_{T,X}^{}, E_{T,Y}^{}, E_{T,Z}^{}}
  \defeq
  m_T^{} \cdot
  \bk*{H\bk*{\mu_T^{}, \mu_T^{}, 1 - 2\mu_T^{}},\; 1,\; 1},
\]
and for a node $T$ of shape $(1,2,1)$, we define
\[
  \bk*{E_{T,X}^{}, E_{T,Y}^{}, E_{T,Z}^{}}
  \defeq
  m_T^{} \cdot
  \bk*{1,\; H\bk*{\mu_T^{}, \mu_T^{}, 1 - 2\mu_T^{}},\;  1},
\]
Then, letting $\mathcal T_2^+$ denote the set of positive-shape nodes at level $2$, the \emph{retained exponent} at level 2 is given by
\[
  E_2^{}
  \defeq
  \min\BK*{
    \sum_{T \in \mathcal T_2^+} E_{T,X}^{}, \;
    \sum_{T \in \mathcal T_2^+} E_{T,Y}^{}, \;
    \sum_{T \in \mathcal T_2^+} E_{T,Z}^{}
  }.
\]

\paragraph{Local matrix size for zero-shape nodes.}

For every zero-shape node $T$, we let $W_0^{} \in \BK{X,Y,Z}$ be the first zero coordinate of $s_T^{}$ and $W_1^{}$ be the first nonzero coordinate of $s_T^{}$.
% and $W'_0 \defeq (Z \to Y \to X \to Z)(W_0)$ be the dimension \emph{before} $W_0$ in the cyclic order.
Then, we define the \defn{local matrix size} of $T$, written $(M_{T,X}^{}, M_{T,Y}^{}, M_{T,Z}^{})$, as
\[
  M_{T, W_0}^{} \defeq
  m_T^{} \cdot \bk[\Bigg]{
  H\bk*{\beta_{T,W_1}^{}}
  + \sum_{L \in \supp(\beta_{T,W_1}^{})}
  \beta_{T,W_1}^{}(L) \cdot
  \left| \myset* {p \in [2^{\ell - 1}]}{L_p = 1} \right|
  \cdot \log q};
\]
the other two coordinates in the local matrix size are zeros.

\paragraph{Local matrix size for $(1,1,2)$-nodes.}
For a node $T$ with shape $(1,1,2)$, we define its local matrix size as
\[
  \bk*{M_{T,X}^{}, M_{T,Y}^{}, M_{T,Z}^{}}
  \defeq
  m_T^{} \cdot \bk*{
    (1 - 2\mu_T^{})\log q, \;
    (1 - 2\mu_T^{})\log q, \;
    2\mu_T^{}\log q
  }.
  \numberthis \label{eq:msize-112}
\]
% The local matrix size for nodes with shapes $(2,1,1)$ and $(1,2,1)$ are obtained by applying the cyclic coordinate permutation
% $X \to Y \to Z \to X$ to \Cref{eq:msize-112} once and twice, respectively.
Similarly, for a node $T$ with shape $(2,1,1)$, we define its local matrix size as
\[
  \bk*{M_{T,X}^{}, M_{T,Y}^{}, M_{T,Z}^{}}
  \defeq
  m_T^{} \cdot \bk*{
    2\mu_T^{}\log q, \;
    (1 - 2\mu_T^{})\log q, \;
    (1 - 2\mu_T^{})\log q
  },
\]
and for a node $T$ with shape $(1,2,1)$, we define its local matrix size as
\[
  \bk*{M_{T,X}^{}, M_{T,Y}^{}, M_{T,Z}^{}}
  \defeq
  m_T^{} \cdot \bk*{
    (1 - 2\mu_T^{})\log q, \;
    2\mu_T^{}\log q, \;
    (1 - 2\mu_T^{})\log q
  }.
\]

\subsection{Final assembly}
Finally, we define the \defn{total retained exponent} as
\[
  E_{\text{total}}^{} \defeq E_G^{} + E_2^{} + \sum_{\ell=3}^{\ell^*} E_\ell^{},
\]
and define the \defn{total matrix size} as
\[
  M_{\text{total}}^{} \defeq \min\BK*{\sum_{T \in \mathcal{L}} M_{T,X}^{}, \; \sum_{T \in \mathcal{L}} M_{T,Y}^{}, \; \sum_{T \in \mathcal{L}} M_{T,Z}^{}},
\]
where $\mathcal{L}$ is the set of leaves. We can now write the full optimization problem.

\begin{tcolorbox}[
    colback=white, % Background color
    colframe=gray!50!black, % Border color
    sharp corners, 
    boxrule=0.5pt,
    left=5pt, right=5pt, top=5pt, bottom=5pt
]
\begin{equation}
  \begin{array}{ll}
    \text{minimize} & \Omega                                                                                       \\[0.2em]
    \text{subject to}
                    & E_{\text{total}}^{} + M_{\text{total}}^{} \cdot \Omega \ge 2^{\ell^*-1} \log(q + 2),               \\
                    & \text{all free variables lie in the domains stated in \cref{subsec:optimizable-parameters}.}
  \end{array}
  \label{eq:optimization-problem}
\end{equation}
\end{tcolorbox}

\begin{theorem}[\cite{alman2025more}]
Any feasible solution of \Cref{eq:optimization-problem} implies $\omega \le \Omega$, where $\omega$ is the asymptotic matrix multiplication exponent.
\end{theorem}

\subsection{Dealing with maximum entropies}
\label{subsec:max-entropy-certification}
Among the quantities introduced above, $H_D^{\max}(\rho)$ is the only type that cannot be computed analytically. However, it is easy to see that replacing each occurrence of $H_D^{\max}$ in the computation of $E_{\text{total}}^{}$ with its upper bound can only decrease $E_{\text{total}}^{}$, so any solution that satisfies the constraints in \Cref{eq:optimization-problem} after this replacement is feasible for \Cref{eq:optimization-problem} itself, and still yields a valid upper bound $\omega \le \Omega$. In this section, we describe how to derive an upper bound of $H_D^{\max}(\rho)$.

For a set $D$ of shapes and a probability distribution $\rho \in \Delta(D)$, a \defn{valid certificate} for $H_D^{\max}(\rho)$ consists of:
\begin{enumerate}
  \item a distribution $y \in \Delta(D)$ satisfying $y(a) > 0$ for every $a \in D$ and $y_W^{} = \rho_W^{}$ for every $W \in \BK{X, Y, Z}$;
  \item a real number $\lambda_0^{}$ and, for each $W \in \BK{X, Y, Z}$ and each marginal value $w \in \myset{a_W^{}}{a \in D}$, a real number $\lambda_W^{}(w)$ (these values play the role of Lagrange multipliers for the maximization defining $H_D^{\max}(\rho)$); and
  \item a real number $\eps \ge 0$ satisfying,
  for every $a = (a_X^{}, a_Y^{}, a_Z^{}) \in D$,
        \[
          \absolute [\big]{\log y(a) - \bk[\big]{\lambda_0^{} + \lambda_X^{}(a_X^{}) + \lambda_Y^{}(a_Y^{}) + \lambda_Z^{}(a_Z^{})}}
          \le \eps. \numberthis \label{eq:certificate-residual}
        \]
\end{enumerate}

\begin{lemma}
  \label{lem:max-entropy-upper-bound}
  For any valid certificate, we have $H(y) \le H_D^{\max}(\rho) \le H(y) + 2 \eps$.
\end{lemma}

\begin{proof}
The left inequality follows from the definition: $y$ has the same marginals as $\rho$, so it is feasible for the maximization defining $H_D^{\max}(\rho)$ (see \Cref{eq:def-hmax}).

For the right inequality, fix any $\rho' \in \Delta(D)$ with $\rho'_W = \rho_W^{}$ for all $W \in \BK{X, Y, Z}$; to prove the right inequality we need only show $H(\rho') \le H(y) + 2 \eps$.

The entropy function $H$ is concave on $\Delta(D)$, and since $y$ is strictly positive, $H$ is differentiable at $y$ with $\nabla H(y)_a = -\log y(a) - \log e$ (recall that all logarithms are in base 2). Therefore
  \[
    H(\rho') \le H(y) + \angbk[\big]{\nabla H(y), \, \rho' - y}.
  \]
  Write $g(a) \defeq \lambda_0^{} + \lambda_X^{}(a_X^{}) + \lambda_Y^{}(a_Y^{}) + \lambda_Z^{}(a_Z^{})$, and view $g$, $\log y$, and the constant $\log e$ as vectors indexed by $a \in D$.
  Then,
  \begin{align*}
    \angbk[\big]{g + \log e, \, \rho' - y}
     & = \bk[\big]{\lambda_0^{} + \log e} \cdot \sum_{a \in D} \bk[\big]{\rho'(a) - y(a)}
    + \sum_{W \in \BK{X, Y, Z}} \, \sum_{a \in D} \lambda_W^{}(a_W^{}) \cdot \bk[\big]{\rho'(a) - y(a)} \\
     & = \bk[\big]{\lambda_0^{} + \log e} \cdot \bk{1 - 1}
    + \sum_{W \in \BK{X, Y, Z}} \, \sum_{w \in \BK{a_W^{} \mid a \in D}} \lambda_W^{}(w) \cdot \bk[\big]{\rho'_W(w) - y_W^{}(w)}
    = 0,
  \end{align*}
  where the last step uses $\rho'_W = \rho_W^{} = y_W^{}$ for each $W \in \BK{X, Y, Z}$. Hence,
  \[
    \angbk[\big]{\nabla H(y), \, \rho' - y}
    = \angbk[\big]{g - \log y, \, \rho' - y}
    \le \eps \cdot \norm{\rho' - y}_1
    \le 2 \eps,
  \]
  where the first inequality uses \Cref{eq:certificate-residual} and the second one uses the fact that the $\ell_1$-distance between two probability distributions is at most $2$. Thus, $H(\rho') \le H(y) + 2 \eps$, proving the lemma.
\end{proof}

\section{Solving the optimization problem numerically}
\label{app:subsubsec:omega_optimization}

We wish to numerically minimize $\Omega$ from \Cref{eq:optimization-problem} while satisfying all the constraints, with $q=5$ and $\ell^*=4$. In \cite{alman2025more}, this was done (for $\ell^*=3$) via a sequential quadratic programming (SQP) algorithm using the software package SNOPT \citep{gill2022snopt}. In contrast, we take a gradient-based approach to tackle the non-convex minimization. We develop an algorithm applying several techniques from machine learning and adjacent areas, such as optimal transport.

\subsection{Differentiable objective}

To apply a gradient-based optimization algorithm, we need a differentiable objective function to be optimized. However, obtaining a differentiable objective from \Cref{eq:optimization-problem} is challenging due to the problem constraints.
Since most of the parameters to be optimized are distributions, this creates the constraints that the probabilities must be non-negative and must add up to one. To sidestep that issue and obtain unconstrained optimization, we parameterize distributions in terms of their logits, and obtain the distribution by applying the softmax operation on the free parameters. We initialize the algorithm with random (logit) parameters.

The other relevant set of constraints comes from the maximum entropy distributions. Given an input probability distribution, we must find the distribution with maximum entropy that shares the same marginals as the input. This is reminiscent of a problem that typically arises in optimal transport, where the Sinkhorn-Knopp algorithm \citep{sinkhorn1967concerning} is used to obtain a stable and differentiable objective \citep{cuturi2013sinkhorn}. Thus, unlike \cite{alman2025more}, we do not treat the maximum entropy distributions (and the corresponding Lagrange multipliers) as free parameters to be optimized together with the rest of distributions; rather, we obtain the maximum entropy distributions (and the Lagrange multipliers) using the Sinkhorn-Knopp algorithm.

We compute the gradients of the objective using automatic differentiation. To improve stability, we employ implicit differentiation for backpropagating through the Sinkhorn-Knopp algorithm \citep{cuturi2020supervised,eisenberger2022unified}. We update the parameters using Adam \citep{kingma2015adam}.

\subsection{Software implementation}

To obtain a fast algorithm that can leverage hardware platforms (such as GPUs) and scale up to level $\ell^*=4$ (with nearly 7 million parameters to optimize) we implement our algorithm in Jax \citep{jax2018github}, applying some techniques to improve efficiency.
The key to an effective speed-up is finding a new representation for the problem, switching from a loop over graph nodes (as done in the SQP implementation from \cite{alman2025more}) to a naturally parallelizable computation over tensors. In the implementation from \cite{alman2025more}, the free parameters are stored in the nodes of a graph where updates happen through message-passing (both from child nodes to parent nodes and vice versa). The graph is non-uniform enough (e.g., different nodes have different number of children and different types of them) to make parallelization hard, but we solve this via two techniques. Firstly, we introduce phantom nodes in the graph with masking; as a trade-off, we pay the cost of increasing the number of optimization parameters by up to a factor of $3$. Secondly, we cluster the nodes into a limited number of highly specialized groups (``stages''). With these techniques, we are able to represent the full graph with multi-dimensional tensors, enabling parallel processing over a number of axes (up to $10$), challenging the limits of tensor-processing backends.

\subsection{Applying \method}

We use \method \citep{novikov2025alphaevolve} to further improve the optimization algorithm. Specifically, we let \method modify the optimization program, which is then executed (taking approximately 5 hours on a single GPU) to output a bound on omega. \method then evolves the code to minimize omega. We found improved results by using \method's ``evolving constructions'' feature, where the optimization algorithm at each generation starts at the best solution point found by the parent algorithm.

\section{Rigorous verification of the omega upper bound}

To rigorously certify our omega bound, we run a separate verification step computing all quantities in rational arithmetic to guard against floating point errors. More specifically, we round the floating point solution obtained at the end of the optimization to rational numbers, ensuring the maximum-entropy certificates stay valid. We then evaluate all derived quantities in exact rational arithmetic, and replace each logarithm with a rational bound rounded in the proper direction that ensures each constraint in \Cref{eq:optimization-problem} is satisfied, so the certified bounds are free from numerical errors. We are preparing a repository in which we will release the verification code and our discovered solution.

\section{Discussion}

In this note, we improve the SOTA upper bound on $\omega$ to 2.371177. 
Our improvement is comparable in magnitude to most improvements in the last 40 years since $\omega < 2.376$ was attained by \cite{coppersmith1990matrix}. 
We achieve this by leveraging modern optimization techniques and \method to design a better optimization algorithm for the problem described by \cite{alman2025more}.
While further modest improvements may be obtained in this manner, achieving larger improvements to $\omega$ likely requires new mathematical ideas and is an exciting area of research.

\paragraph{Acknowledgments.} The work presented in this note was motivated by initial conversations during the Complexity and Linear Algebra program at Simons Institute for the Theory of Computing in Fall 2025.

\newpage

\bibliography{main}

\end{document}